\documentclass[%
reprint,
superscriptaddress,
frontmatterverbose, 
preprintnumbers,
nofootinbib,
 amsmath,amssymb,
 aps,
 pra,
notitlepage
]{revtex4-1}

\usepackage{pdfpages}
\usepackage{pgffor} 

\usepackage{graphicx}
\usepackage{dcolumn}
\usepackage{bm}
\usepackage{hyperref}
\usepackage{subfigure}
\hypersetup{
colorlinks=true,
linkcolor=blue,
filecolor=blue,
citecolor=blue,  
urlcolor=black,
}

\usepackage{amsthm}
\usepackage{enumerate}
\usepackage{comment}
\usepackage{tikz}

\usepackage{url}

\DeclareMathOperator{\Tr}{Tr}
\newcommand{\E}{\mathbb{E}}

\def\>{\rangle}
\def\<{\langle}

\newcommand{\OTOC}{\rm{OTOC}}
\newcommand{\ent}{\rm{ent}}

\makeatletter
\AtBeginDocument{\let\LS@rot\@undefined}
\makeatother 

\newtheorem{theo}{Theorem}

\newtheorem{lemma}{Lemma}

\newtheorem{cor}{Corollary}
\newtheorem{defi}{Definition}

\newcommand{\nc}{\newcommand}
\nc{\rnc}{\renewcommand}
\nc\benum{\begin{enumerate}}
\nc\eenum{\end{enumerate}}
\nc\bit{\begin{itemize}}
\nc\eit{\end{itemize}}

\newcommand{\thmref}[1]{Theorem~\ref{thm:#1}}

\newcommand{\corref}[1]{Corollary~\ref{cor:#1}}

\begin{document}
\title{A Separation of Out-of-time-ordered Correlation and Entanglement}

\newcommand{\cC}{\mathcal{C}}
\author{Aram W. Harrow} \email{aram@mit.edu}
\author{Linghang Kong} \email{linghang@mit.edu}
\affiliation{Massachusetts Institute of Technology}
\author{Zi-Wen Liu} \email{zliu1@perimeterinstitute.ca}
\affiliation{Perimeter Institute for Theoretical Physics}
\author{Saeed Mehraban} \email{mehraban@mit.edu}
\affiliation{Massachusetts Institute of Technology}
\author{Peter W. Shor} \email{shor@math.mit.edu}
\affiliation{Massachusetts Institute of Technology}

\begin{abstract}
The out-of-time-ordered correlation (OTOC) and entanglement are two physically motivated and widely used probes of the ``scrambling'' of quantum information, a phenomenon that has drawn great interest recently in quantum gravity and many-body physics.
We argue that the corresponding notions of scrambling can be fundamentally different, by proving an asymptotic separation between the time scales of the saturation of OTOC and that of entanglement entropy
in a random quantum circuit model defined on graphs with a tight
bottleneck, such as tree graphs. 
Our result counters the intuition that a random quantum circuit mixes in time proportional to the diameter of the underlying graph of interactions. 
It also provides a more rigorous justification for an argument in our
previous work~\cite{Shor18}, that black holes may be slow information
scramblers, which in turn relates to the black hole information
problem.
The bounds we obtained for OTOC are interesting in their own right in that they generalize previous studies of OTOC on lattices to the geometries on graphs in a rigorous and general fashion.
 
\end{abstract}

\preprint{MIT-CTP/5131}
  
\maketitle


\section{Introduction and overview}


The ``scrambling'' of quantum information is a phenomenon of fundamental importance, deeply connected to many important research topics in physics, such as black holes \cite{HP07,SS08,MSS16,SS14} and many-body chaos \cite{NH15,PH10}. 
In recent years, a great amount of research effort has been devoted to the detection and characterization of scrambling. 
The so-called out-of-time-ordered correlation (OTOC) \cite{LO69} is a commonly used measure of quantum chaos and scrambling. A variant based on commutators (also known as the OTO commutator) is given by
\begin{equation}
 C(t) = \frac{1}{2}\<[O_1(x,0), O_2(y,t)]^\dagger [O_1(x,0),O_2(y,t)]\>, \label{eq:otoc}
\end{equation}
 where $O_1(x,0)$ is an operator acting on site $x$, and $O_2(y,t)$ is a Heisenberg operator at time $t$ that only acts on $y$ at time 0, i.e.\ $O_2(y,t) = U^\dagger(t) O_2(y,0) U(t)$ where $U(t)$ is the unitary for the evolution from time 0 to $t$. The average is taken w.r.t.~the thermal state at some temperature, which we take to be infinite in this work. Intuitively speaking, it characterizes parameters like sensitivity to initial conditions via the spread of local operators.  
Also notice that the scrambling phenomena exhibit a truly quantum nature---the state of the entire system remains pure during the unitary evolution (although it is effectively randomized), thus no information is really lost; the generation of global entanglement leads to the scrambling of initially localized quantum information, spreading and hiding it from observers that only have access to part of the system.  
This observation leads to another fundamental probe of a stronger form of scrambling, namely the entanglement between parts of the system \cite{SS08,LSH+13,HQRY16,LLZZ18} (which measures the equivalent effect as the tripartite information \cite{HQRY16} in the case of unitary dynamics; see \cite{LLZZ18} for more detailed discussions).

To understand and characterize the dynamical behaviors of scrambling systems, several
explicit models have been proposed and investigated, such as the Sachdev-Ye-Kitaev (SYK)
model~\cite{PhysRevLett.70.3339,k}.  Another leading approach is the random quantum
circuit model \cite{HP07,HQRY16, NVH18, NVH17, PhysRevX.8.031057}, capturing the key
kinematic feature of chaos that the evolution appears to be random, and the
locality of physical interactions. 
In these well-studied physical scrambling models, the saturation of OTOC and that of entanglement are expected to occur at a similar time scale \cite{NVH18, NVH17,HQRY16}.
More generally, one could consider the dynamics of many small quantum systems (say qubits) connected according to some graph \cite{HK06,NS06}, with random unitary gates being applied to each edge.  Suppose that we apply gates in a random order such that on average each edge has one gate applied to it per unit time. 
A natural conjecture here, which would be compatible with all previous results, is that
the scrambling time for a constant degree graph is proportional to its diameter, i.e.\ the maximum distance between any two vertices.  This would correspond to information traveling through the graph at a linear velocity and is assumed implicitly in previous works.
However, no proof exists, outside of the special case of Euclidean lattices in a fixed number of dimensions.  Even for Euclidean lattices in more than one dimension, this result was only recently proven \cite{HM18}.


 Our main results are the following.
    The first one (\thmref{otoc-mixing} and \thmref{otoc-lb}) shows that for arbitrary graphs with sufficiently low degree, the OTOC saturation time scales linearly in the graph diameter.  Here by low degree, we mean $d^2 \ge z$ where $d$ is the dimension of the quantum system and $z$ is the degree of the graph.
  On the other hand, we use bounds on entanglement growth to show that the time needed to establish substantial entanglement between parts of the system scale at least as the number of vertices and thus could be longer than the OTOC saturation time, for graphs with bottlenecks (see \corref{ent}).  Such graphs include e.g.~binary trees, which we explicitly analyze in this paper, and discretizations of hyperbolic space around black holes, originally proposed by \cite{Shor18}, which are expected to exhibit similar behaviors (as argued below). In other words, we have established an \emph{asymptotic} separation between the time scales of OTOC and entanglement saturation.  Refs.~\cite{BGL18,Lu19} studied scrambling on certain peculiar graphs via a Hamiltonian model, but the relations between OTOC and entanglement were not fully understood and the physical correspondences were not clear; here we rigorously proved the separation in a general setting and studied the implications to physics. 
  

Our results have the following major implications:

i) {\em Scrambling in non-Euclidean geometries.}
Existing work studied scrambling mostly on Euclidean lattices \cite{NVH18, NVH17,JHN18}. The general assumption is that after time $t$, a localized perturbation will affect everything within some ball of radius $v_{\text{butterfly}}t$, where $v_{\text{butterfly}}$ is known as the ``butterfly velocity'' which characterizes the speed of information spreading. However, this has not been proved and previous works only gave heuristic arguments for it that included uncontrolled approximations.
For the random circuit models defined on general graphs, we find that if the local dimension
is large relative to the graph degree then indeed there is a linear butterfly velocity.
We also find apparent counter-examples which suggest that linear butterfly velocity no longer holds in high-degree graphs. Some of these examples are not rigorously analyzed but we present a heuristic argument suggesting that the scrambling time for some families of graphs should grow more rapidly or more slowly than the diameter of the graphs. 
    
    ii) {\em Black hole information scrambling.}
Our results can be regarded as a more rigorous argument that fleshes out the idea of a recent paper by one of the authors \cite{Shor18}, which concerns whether it is possible for the fast scrambling conjecture of black holes~\cite{SS08}
to hold if one assumes that the causality structure of general relativity 
holds around a black hole, and if the medium by which the information is 
scrambled is Hawking radiation. In the model of \cite{Shor18}, the space around the black hole is divided 
into cells, each of which contains a constant number of bits of Hawking 
radiation. 
It then gives arguments for why the Hawking radiation is not 
adequate for fast scrambling if the entanglement definition of scrambling 
is used. 
The cell structure around the black hole looks like a patch of a cellulation of hyperbolic geometry, where the cells on the event horizon 
are the boundary of this patch.  
The tree graph we shall consider captures a key feature
  of this geometry: the leaves lie on the event horizon, and the density of nodes decreases as one moves outwards radially.
As the assumptions essentially suggest that information is processed via local interactions of the Hawking radiation, we may consider a random circuit defined on the underlying graph to be a toy model that captures key features of the black hole scrambling process. Our mathematical results then indicate that the scrambling time scales given by entanglement and OTOC are fundamentally different.
Another way to interpret our model is that the information ``wavefront'' could reach the farthest side rather quickly since there exists short paths, but it takes a longer time, which scales with the number of degrees of freedom, to establish truly global entanglement. This is consistent with
recent holographic calculations (see
e.g.~\cite{LiuSuh14:tsunami,CENSX19}), which suggest that the
entanglement entropy grows roughly linearly after a quench in chaotic systems. 

We would also like to remark upon the task of recovering quantum information falling into the black hole from Hawking radiation (commonly known as Hayden-Preskill decoding \cite{HP07}), which plays central roles in recent studies of the black hole information problem. Yoshida and Kitaev recently proposed an explicit protocol \cite{2017arXiv171003363Y} whose decoding fidelity is at least the order of $1/{d_A^2 (1-C(t))}$,
where $d_A$ is the Hilbert space dimension of the input message.
Here $C(t)$ takes the form of Eq.~\eqref{eq:otoc} and considers $O_1$ and $O_2$ averaged
over all Pauli operators on the infalling system and Hawking radiation respectively; see
Sections 2--4 of \cite{2017arXiv171003363Y} for details.
By simple calculations one can see that our results imply a possible
time window in which the decoding could be achieved with high fidelity
without substantial entanglement when the infalling quantum state
is sufficiently small compared to the black hole. However, it
appears that adding a small number of qubits to a Schwarzschild black hole can only
be done by photons whose wavelength is comparable to the size of the
black hole. It does not seem surprising that the information carried
by such photons can be extracted by a black hole quickly; when the
information is absorbed by the black hole, it is already spread out over
the entire black hole, and so does not need to migrate from a localized
region to a state where it is delocalized on the black hole.

iii) {\em Inequivalence of convergence to 2-designs in different measures.}
The speed of convergence of a random circuit to a 2-design (distributions that
approximately agree with the Haar measure up to the first two moments, which have found
many important applications as an efficient approximation to Haar randomness \cite{L10}) has been the subject of a large amount of research.   In
particular, \cite{HL09, HM18, BHH16, BF12, BF15} show that the speed of convergence
depends on the graph of interactions, and suggest that it should be proportional to the
diameter.  Note that 2-designs are very powerful measures of
convergence, in the sense that a distribution being close to a 2-design implies that the
distribution has mixed with respect to not only OTOC but also von Neumann and R\'enyi-2
entanglement entropies \cite{HQRY16,PhysRevLett.120.130502}, and other important signatures
of information scrambling such as decoupling \cite{1367-2630-15-5-053022}.  
Our work provides several
examples where a random circuit approximates the OTOC but not the entanglement properties of a 2-design, and therefore implies that a strong approximation of 2-designs (in terms of e.g.~the frame operator \cite{LLZZ18}) may not be achieved in time proportional to the diameter.


\section{Models and Notation}
\label{sec:mc}


Let $G$ be a graph with $V$ vertices and $E$ edges. The model we study consists of a graph
with a $d$-dimensional Hilbert space associated with each vertex of $G$. Each edge has
Haar-random unitary gates applied to qudits on its endpoints according to a Poisson process
with rate 1, meaning a Poisson distribution such that $k$ unitaries are applied in time $t$ with probability $t^ke^{-t}/k!$). The mixing times for OTOC and entanglement,
$\tau_{\OTOC}^{(x,y)}$ and $\tau_{\ent}^{(A)}$, are defined as follows.
\begin{defi}
  $\tau_{\OTOC}^{(x,y)}$ (resp.~$\tau_{\ent}^{(A)}$) is defined to be the minimum amount
  of time needed for OTOC between vertices $x$ and $y$ (resp.~the entanglement entropy
  between $A$ and the complement of $A$) to become at least a constant fraction of its equilibrium
  value. In this work we take the constant to be $1/(d^2+1)$.
\end{defi}
Here we expect that qualitatively similar behavior will hold with $1/(d^2+1)$ replaced by any
constant strictly between 0 and 1.  We will study how $\tau_{\OTOC}$ and $\tau_{\ent}$
scale with parameters such as local dimension, degree, and number of vertices.  

We study the pair of $(x,y)$ that has largest
$\tau_{\OTOC}^{(x,y)}$, and the set $A$ that has largest $\tau_{\ent}^{(A)}$,
as they could best characterize OTOC and entanglement properties for $G$.  
  
Instead of studying this model directly we can consider the process in which a random edge is picked every $1/E$ time units. This is because in our Poisson process model, each edge is equally likely to be picked.  The number of unitaries applied within time $t$ is of order $Et$ (see Appendix A), so the two models above are equivalent up to a constant factor.


\section{OTOC}
\label{sec:otoc}

To analyze the saturation time of OTOC, we describe the process of operator spreading as a Markov chain.
Consider an arbitrary Pauli operator $\sigma_{\vec p}$ acting on $n$ $d$-dimensional qudits, $\vec p \in \{0,\ldots,d^2-1\}^n$, and apply some unitary $U$ to it. We expand the resulting operator on Pauli basis and have
\[
 U^\dagger \sigma_{\vec p} U = \sum_{\vec q} \alpha_{\vec q} \sigma_{\vec q}, \quad
 \alpha_{\vec q} \equiv \frac{1}{d^n}\Tr[U^\dagger\sigma_{\vec p} U \sigma_{\vec q}^\dagger].
\]
The expected value of the cross term for $\alpha_{\vec q}$ averaged over the distribution of $U$ would be
\[
 \E_U\alpha_{\vec q} \alpha_{\vec {q'}}^* = \frac{1}{d^{2n}} \E_U\Tr[U^{\dagger\otimes 2} (\sigma_{\vec p} \otimes \sigma_{\vec p}^\dagger) U^{\otimes 2}(\sigma_{\vec q}^\dagger \otimes \sigma_{\vec {q'}})].
\]

According to the construction of random circuit, this is zero for $\vec q \not= \vec {q'}$ for $U$ being the unitary in a single step. Therefore in each step the values of $\overline{\alpha_{\vec q} \alpha_{\vec q}^*}$ undergo linear transformation, which we can interpret as a distribution because they are positive and sum to 1.

If we start from a Pauli operator located at a single vertex $x$, on each vertex all non-identity Pauli operators will have the same probability as long as $x$ has been touched at least once in the process. So we only care if the operator on a vertex is identity (I) or non-identity (N). And the norm of the time-evolved operator with a Pauli operator $P$ on some vertex $y$ would be proportional to the probability of having nonzero Pauli operator on that site, and the factor of proportionality would be
\[
 \frac{1}{d^2-1}\sum_{i=1}^{d^2-1}\frac{1}{2d} \Tr([\sigma_i,P]^\dagger [\sigma_i,P]) = \frac{d^2}{d^2-1},
\]
which is just the commutator averaged over all non-identity Pauli operators.

In summary, the object we will study is the OTOC between Pauli operator on vertex $y$ and time-evolved Pauli operator on vertex $x$ after $T$ steps of random circuits on graph $G$, which equals to $\frac{d^2}{d^2-1}$ times the probability of having ``N'' on vertex $y$ after $T$ steps in the Markov chain $M_0$ defined below.
The state space of $M_0$ is the set of all the configurations in which each vertex of $G$ is assigned a label ``N'' or ``I''. The initial state of $M_0$ has ``N'' assigned to vertex $x$ and ``I'' assigned to all other vertices. The update rule is that in each step a uniformly random edge is picked and the labels on the two corresponding vertices are updated. ``II'' remains ``II'', and otherwise they has a probability of $\frac{d^2-1}{d^4-1} = \frac{1}{d^2+1}$ for becoming ``IN'' or ``NI'' each, and $\frac{d^2-1}{d^2+1}$ for becoming ``NN'' \cite{HL09}.

Now we prove an upper bound for the OTOC saturation time. For illustration we present an outline of the proof here; full proof could be found in Appendix C. Note that $O(\alpha)$ that appears here and in the following is represents a quantity that scales asymptotically as $\alpha$, i.e.~$\geq c_1 \alpha$ and $\leq c_2\alpha$ for constants $0<c_1<c_2$.
\begin{theo}[OTOC upper bound]
Let $G$ be a graph with $V$ vertices and $E$ edges, and suppose the degree for each vertex at most $d^2$, where $d$ is the Hilbert space dimension for each vertex. Then for any pair of vertices $x$ and $y$, $\tau_{\OTOC}^{(x,y)} = O (D(x,y))$ with high probability, where $D(x,y)$ is the distance between $x$ and $y$. The probability of failure is exponentially small in $D(x,y)$. As a consequence the perfect binary tree has $\tau_{\OTOC}^{(x,y)} = O (\ln V)$, where $x$ and $y$ are the farthest pair of vertices.
\label{thm:otoc-mixing}
\end{theo}
\begin{proof}
 As explained earlier, the OTOC saturation time corresponds to the number of steps needed for $M_0$ to have constant probability of having a label ``N'' on $y$.
 We will first prove Lemma~\ref{thm:otoc}, which states that with probability $1-e^{-O(D(x,y))}$ the vertex $y$ gets hit by a label ``N'' within order of $E\cdot D(x,y)$ steps. As shown in Appendix~A, this needs order of $D(x,y)$ time units with high probability. Then we will show in Lemma~\ref{thm:const} that after this happens, the probability for having an ``N'' on $y$ remains constant.
\end{proof}

\begin{lemma}
 \label{thm:otoc}
 Suppose that $G$ is a graph with the degree for each vertex being at most $d^2$. For any pair of vertices $x$ and $y$ with distance $D(x,y)$, the expected number of steps for $y$ to be labeled ``N'' is of order $E\cdot D(x,y)$ in $M_0$ starting from $x$. Besides, with high probability the vertex $y$ gets labeled ``N'' in time of order $E\cdot D(x,y)$. The probability of failure is exponentially small in $D(x,y)$.
\end{lemma}
\begin{proof}
  (Sketch) We will first construct a Markov chain $M$ which has the same initial state as $M_0$, and in each step the update rule of $M$ is applied, followed by changing all ``N'' into ``I'' except the one closest to vertex $y$. By a simple coupling argument the number of steps needed for $y$ to get an ``N'' in $M_0$ is lower bounded by that in $M$. The distance between the vertex with label ``N'' and vertex $y$ in Markov chain $M_0$ can be described by a biased random walk, from which we can obtain the desired bound. More details could be found in Appendix~C.
\end{proof}

\begin{lemma}
 \label{thm:const}
 After a label ``N'' reaches the target vertex $y$, the probability for having an ``N'' on $y$ will remain order one. 
\end{lemma}
\begin{proof}
 (Sketch) We again consider the modified chain which only keeps one label ``N'' after each step. We will show that vertex $y$ has constant probability of having label ``N'' in the equilibrium distribution. This probability is monotonically non-increasing as a function of the number of steps, so the probability is order one in any step. More details could be found in Appendix~D.
\end{proof}

Theorem~\ref{thm:otoc-mixing} states that the number of steps needed for OTOC saturation in a low-degree graph is at most of order $E\cdot D(x,y)$. However, we expect that in a graph with high degree, the number could be much larger.  Some intuitions are given in Appendix~E. 

Besides this upper bound we also derive a lower bound for OTOC saturation.
\begin{theo}[OTOC lower bound]
Let $G$ be a graph with $V$ vertices and $E$ edges, and suppose the degree for each vertex is $O(1)$. Then for any pair of vertices $x$ and $y$, $\tau_{\OTOC}^{(x,y)}$ is at least of order $D(x,y)$ with high probability, where $D(x,y)$ is the distance between $x$ and $y$. The probability of failure is exponentially small in $D(x,y)$.
\label{thm:otoc-lb}
\end{theo}
The proof is presented in Appendix~F.

\section{Entanglement}
\label{sec:ent}
Here we only need to consider the case where the evolution is unitary and the system is pure.

Entanglement entropy of pure state $|\psi\>_{AB}$ is given by $E(|\psi \>) := S(\rho_A )$ where $\rho_A = \Tr_B[ |\psi \>\<\psi |]$ and $S$ is
the von Neumann entropy.
Notice the following simple, general fact:

\begin{lemma}
\label{lem:ent}
 Let $U_{AB}$ be a unitary operator acting on two $d$-dimensional systems $AB$. Then for any $|\psi\>_{AA'BB'}$ with ancilla systems $A' $, $B' $,
\[
 E((U_{AB}\otimes \operatorname{id}_{A'B'})|\psi\>_{AA'BB'}) - E(|\psi\>_{AA'BB'}) \le 2\log d.
\]

\end{lemma}
\begin{proof}
  Adapted from the proof of Lemma 1 of \cite{BHLS03}. Suppose Alice holds $AA'$ and Bob holds $BB'$. In addition, they share two copies of the maximally entangled state $|\Phi_d\>  =
\frac{1}{\sqrt d} \sum_{i=1}^d|i\>|i\>$, $E(|\Phi_d \>) = \log d$. Consider the following double teleportation protocol. Alice consumes a $|\Phi_d\>$ and classical communication to teleport A to Bob, who performs $U$ locally and then consumes a $|\Phi_d\>$ and classical communication to teleport system $A$ back to Alice. The protocol is LOCC, under which the
entanglement entropy between Alice and Bob is monotonically nonincreasing. Therefore, by the additivity of $S$ (and thus $E$) on tensor products,
\begin{align*}
  &E(|\psi\>_{AA'BB'}) +2E(|\Phi_d\>) \ge  E((U_{AB}\otimes \operatorname{id}_{A'B'})|\psi\>_{AA'BB'}),
\end{align*}
and so the claimed bound follows.
\end{proof}

Note that the proof also applies to e.g.~the R\'enyi-2 entropy, which is a variant of the entanglement entropy that can be more easily measured in experiments \cite{Islam2015,2018arXiv180605747B}.   

By Lemma~\ref{lem:ent}, the entanglement entropy between the two trees increases by at
most $2\log d$ when the random unitary is acted across the middle edge. This edge only has
a probability of $1/E \sim 1/V$ to be selected in each step. So in order to reach the
maximum entropy of order $V\log d$, we need at least an order of $V^2$ steps or equivalently order $V$
time.  This is much larger than the OTOC time of order $\log V$.

From Lemma~\ref{lem:ent} we can get the following result for a general graph.
\begin{cor}\label{cor:ent}
 For a general graph $G$ with vertices partitioned into sets $A$ and $B$, the expected entanglement saturation time is at least of order $ \frac{\min\{|A|,|B|\}}{C(A,B)}$, where ${C(A,B)}$ is the number of edges with one endpoint in $A$ and one in $B$.
\end{cor}
This also lower-bounds the time it takes for the random circuit to converge to 2-designs \cite{HQRY16,PhysRevLett.120.130502}.

\section{OTOC vs.~entanglement in different models}
\label{sec:models}

 
According to our results, a simple graph that can give a separation of OTOC and entanglement saturation times is a perfect binary tree, as depicted in Fig.~\ref{fig:tree}. Here we consider the OTOC of operators located on the pair of farthest vertices in the graph (a leaf vertex of left subtree and a leaf vertex of the right subtree), and the entanglement entropy between left and right subtrees.
The cell structure roughly equivalent to the hyperbolic geometry in 3 dimensions, or indeed any
constant number of dimensions, exhibits such a separation as well. These graphs are regarded as toy models of quantum information scrambling around black holes \cite{Shor18}, as motivated in the introduction; see also Appendix B. 

The behaviors of OTOC and entanglement on some other graphs are also studied.  
These include: i) The ``dumbbell graph'' consisting of two complete graphs connected by a bottleneck edge. A careful analysis could still show a separation between OTOC and entanglement; ii) High degree graphs. We demonstrate two examples in which the OTOC saturation time can be much longer or shorter than the bound we have on low degree graphs. See Appendices~E, G for more details.
Also note that there is no separation on Euclidean lattices.
These results are summarized in Table~\ref{tab:compare}.


\begin{table}[ht]
\caption{\label{tab:compare}Comparison of the time scales of OTOC and entanglement saturation in various graphs. $n$ is the total number of vertices in the graph.}
\begin{ruledtabular}
\begin{tabular}{lll}
Models/graphs & OTOC & Entanglement \\ \hline
Euclidean lattices in $D$ dimensions & $n^{1/D}$ & $n^{1/D}$\\
Hyperbolic space with $D=3$& $\log{n}$ & $\sqrt{n}$ \footnote{See Section 4 of \cite{Shor18}.} \\
Binary tree & $\log n$ & $n$\\
Tree with degree $z\gg d^2$ & $n^{1-\frac{\log d^2}{\log z}}$ & $n/z$ \\
Dumbbell graph & $\log n/n $ & $n$
\end{tabular}
\end{ruledtabular}
\end{table}

Other than the Poisson process on each edge, different orders of choosing the edges have also been studied in Appendix~G.

\begin{figure}[ht]
 \includegraphics[width=0.25\textwidth]{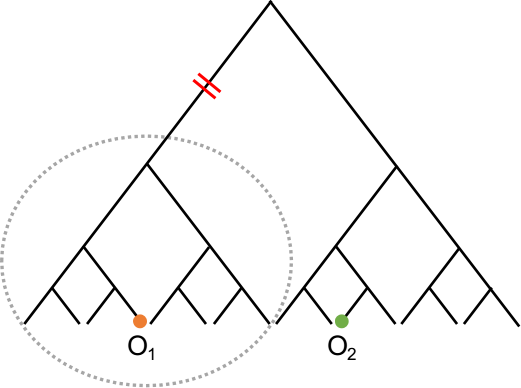}
 \caption{An illustration of the binary tree model of depth 4. We consider OTOC between local operators originally acting on two farthest vertices (a leaf of the left and right subtrees respectively, for example, $O_1$ and $O_2$ in the diagram), and entanglement between the left subtree (dashed circle) and the rest of the graph (the cut shown by the red double line).  \label{fig:tree}}
\end{figure}

\section{Conclusion}

Random quantum circuits have widespread applications in quantum information, and
are also very important models of scrambling and chaotic quantum systems in
theoretical physics. There are several ways to characterize scrambling and
randomness in quantum processes, among which the OTOC and entanglement are two
important types of measures.  This work aims to understand whether they are
equivalent to each other as the signature of scrambling.    To this end, we
carefully analyze local random quantum circuits defined on, e.g., a binary tree, which
exhibit the property that OTOC mixes rather fast since the light cone can
quickly reach the far end (time of order $\ln V$), while it takes a much longer time
for entanglement between the left and right subtrees to grow (time at least of order $V$).  
We furthermore generalize the result to any bounded-degree graph with a tight
bottleneck.  That is, the generation of entanglement is slow, even if the graph
has small diameter.  Our result indicates that unitary $t$-designs can be much
more expensive than we thought: They require a random quantum circuit to have
depth much larger than the diameter of the underlying graph.   This result
provides a more rigorous evidence for arguments made in \cite{Shor18}:  if we
consider the model discussed in \cite{Shor18, SS08}, then the scrambling of
quantum information as seen by strong measures such as entanglement or
decoupling can be much slower than we thought before.
It would be interesting to explore further implications of this phenomenon, and more generally, quantum information processing, to the black hole information problem, many-body physics, and beyond.

\begin{acknowledgments}
AWH was funded by NSF grants CCF-1452616, CCF-1729369, and PHY-1818914; ARO contract
W911NF-17-1-0433; and the MIT-IBM Watson AI Lab under the project Machine Learning in
Hilbert space. ZWL is supported by Perimeter Institute for Theoretical Physics. Research
at Perimeter Institute is supported by the Government of Canada through Industry Canada
and by the Province of Ontario through the Ministry of Research and Innovation.  LK was
funded by NSF grant CCF-1452616. SM was funded by NSF grant CCF-1729369. PWS is supported
by the National Science Foundation under grants No. CCF-1525130 and  CCF-1729369 and through the NSF Science and Technology Center for Science of Information under Grant No. CCF-0939370.
\end{acknowledgments}

\bibliographystyle{unsrt}



\section*{Supplementary material}

# Supplemental Material: A Separation of Out-of-time-ordered Correlator and Entanglement

Aram W. Harrow[*] and Linghang Kong[†]
*MIT Center for Theoretical Physics*

Zi-Wen Liu[‡]
*Perimeter Institute for Theoretical Physics*

Saeed Mehraban[§]
*MIT Department of Electrical Engineering & Computer Science*

Peter W. Shor[¶]
*MIT Department of Mathematics*

### Appendix A: Poisson clock

Here we argue that the Poisson clock model and the simplified model in which a random edge is picked every $1/E$ time units are essentially equivalent, by the following fact:

**Proposition 1.** *Within time t, the number of unitaries applied in Poisson clock model is between $\frac{1}{2}Et$ and $\frac{3}{2}Et$ with probability $1 - e^{-\Omega(Et)}$.*

*Proof.* The number of unitaries applied to each edge is a Poisson distribution with mean $t$, so the total number of unitaries $\lambda$ is a Poisson distribution with mean $Et$. By [1],

$$\Pr\left[|\lambda - Et| \geq \frac{1}{2}Et\right] \leq 2\exp\left(-\frac{(\frac{1}{2}Et)^2}{2(Et + \frac{1}{2}Et)}\right)$$
$$= e^{-\Theta(Et)}.$$

□

### Appendix B: Cellulation of hyperbolic geometry

Here we include the figure from [2] that roughly depicts the cellulation of hyperbolic geometry, representing the causality structure of a black hole in Schwarzschild coordinates.

---

[*] aram@mit.edu
[†] linghang@mit.edu
[‡] zliu1@perimeterinstitute.ca
[§] mehraban@mit.edu
[¶] shor@math.mit.edu

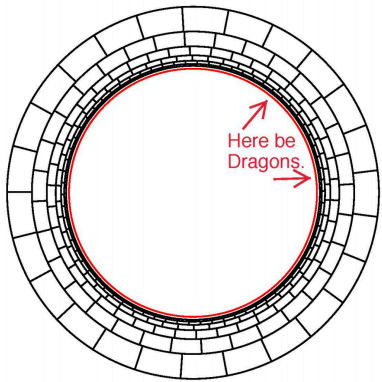

FIG. 1. ([2]) The cell structure roughly equivalent to the hyperbolic geometry, depicted in two dimensions. In [2], this represents the black hole cell structure in Schwarzschild coordinates, where each cell carries one qubit of Hawking/Unruh radiation.

The line element in this geometry is

$$ds^2 = -\left(1 - \frac{2M}{r}\right)dt^2 + \left(1 - \frac{2M}{r}\right)^{-1}dr^2 + r^2(d\theta^2 + \sin^2\theta d\phi^2)$$

The black hole has radius $2M$, the photon sphere has radius $3M$, and we are interested in the region $2M \leq r \leq 3M$. The total number of cells $N$ is $O(M^2)$. A great circle of radius $r = 2M + h$ traverses $\sim \sqrt{M/h}$ cells while a line from radius $2M + h$ to radius $3M$ traverses $\sim \log(M/h)$ cells. Thus the graph diameter is $O(\log N)$.

### Appendix C: Proof of Theorem 1

In order to give an upper bound on the time scale that an "N" hits vertex $y$, we study a modified chain $M$ in which the spreading of label "N" is slower than $M_0$. Roughly speaking $M$ only keeps a single label "N" that is closest to vertex $y$.

**Definition 1** (Markov chain $M$). *Markov chain $M$ has the same state space and initial state as $M_0$. In each step the update rule for $M_0$ is applied, followed by setting all "N" labels into "I" except for the one closest to $y$ (choose randomly if this is not unique).*

In this way the vertex with label "N" in $M$ is always labeled "N" in $M_0$ in the most natural way of coupling $M_0$ to $M$, and therefore after any number of steps the probability that vertex is labeled "N" in $M_0$ is lower bounded by the corresponding probability in $M$.

The Bernstein inequality is needed for the proof of our theorem, which states that for independent zero-mean random variables $X_1, \ldots, X_n$ each with absolute value at most $M$,

$$\Pr\left[\sum_{i=1}^n X_i > t\right] \leq \exp\left(-\frac{\frac{1}{2}t^2}{\sum \mathbb{E}[X_i^2] + \frac{1}{3}Mt}\right).$$

This could be generalized to the case with nonzero mean. Suppose $Y_1, \ldots, Y_n$ has mean $\mu_1, \ldots, \mu_n$ and they satisfy $|Y_i - \mu_i| \leq M$ then by setting $X_i = Y_i - \mu_i$ we have

$$\Pr\left[\sum_{i=1}^n Y_i > t + \sum_{i=1}^n \mu_i\right] \leq \exp\left(-\frac{\frac{1}{2}t^2}{\sum(\mathbb{E}[Y_i^2] + \mu_i^2) + \frac{1}{3}Mt}\right). \tag{C1}$$

**Proposition 2** (Restatement of Lemma 1). *Suppose that $G$ is a graph with the degree for each vertex being at most $d^2$. For any pair of vertices $x$ and $y$ with distance $D(x,y)$, the expected number of steps for $y$ to be labeled "N" is $O(ED(x,y))$ in $M_0$ starting from $x$. Besides, with probability $1 - e^{-\Omega(D(x,y))}$ the vertex $y$ gets labeled "N" in time $O(ED(x,y))$.*

*Proof.* As mentioned, the upper bound for $M$ defined in Definition 1 gives an upper bound for $M_0$. Let $v$ be the vertex with label "N". As long as $v \neq y$, there will be at least one neighbor $u$ that is one step closer to $y$, and other

neighbors are at most one step further from $y$ due to the triangle inequality. If the edge $(u,v)$ is selected, there is a chance of $\frac{d^2}{d^2+1}$ that $u$ obtains label "N". If the edge between $u$ and other neighbors is selected, there is a chance of $\frac{1}{d^2+1}$ that the label on $u$ becomes "I", and the distance between $y$ and the closest label "N" becomes one step longer. Let the degree of $u$ be $d_u$, and the distance between the label "N" and vertex $y$ will have probability of at least $\frac{1}{E}\frac{d^2}{d^2+1}$ to decreases by 1 and at most $\frac{d_u-1}{E}\frac{1}{d^2+1}$ to increase by 1, where the probabilities depend on the specific vertex. Since the degree for any vertex is at most $d^2$, the time needed for the distance to drop from $D(x,y)$ to 0 is upper bounded by the time in the following biased random walk $W$. $W$ has states $\{0,1,\ldots,d_{\max}\}$ where $d_{\max}$ is the maximum possible distance to $y$, and starting from vertex $D(x,y)$ it has a fixed probability of $\frac{d^2}{E(d^2+1)}$ for decreasing by 1 and $\frac{d^2-1}{E(d^2+1)}$ for increasing by 1.

Extension of this finite chain to an infinite one could only increase the hitting time of vertex 0, because the finiteness at vertex $d_{\max}$ prevents us from getting too far from vertex 0. The displacement of an random walk on an infinite chain (i.e. the difference of the final position and initial position) is the sum of displacement for each step, which has probability $\frac{d^2}{E(d^2+1)}$ of being $-1$ and $\frac{d^2-1}{E(d^2+1)}$ probability of being $+1$, and otherwise it is 0. The mean and variance for displacement at each step is

$$\mu_0 = -\frac{1}{E(d^2+1)}, \quad \sigma_0^2 = \frac{2d^2+1}{E(d^2+1)} - \mu_0^2.$$

The expected number of steps needed to reach vertex 0 in this random walk is $\frac{-D(x,y)}{\mu_0} = \Theta(E \cdot D(x,y))$. We can also use Eq. (C1) to bound the probability that the total displacement of $T$ steps is larger than $-D(x,y)$, where we set $T$ to be twice the expected number of steps needed and $t = D(x,y)$. $M$ can be set to be 2. The denominator in the exponent will be $T(\sigma_0^2 + 2\mu_0^2) + \frac{2}{3}t = \Theta(D(x,y))$, so Eq. (C1) gives a probability of at most $e^{-\Theta(D(x,y))}$ for not reaching vertex 0. □

## Appendix D: Proof of Lemma 2

**Lemma 1.** *Consider a reversible Markov chain $M$ with transition matrix $P(x,y)$, $x,y \in \Omega$. $M$ starts deterministically from state $x_0 \in \Omega$. If all the eigenvalues of $P$ are nonnegative, then the probability for $x_0$ will be monotonically non-increasing as a function of the number of steps.*

*Proof.* Let $\pi(x)$ be the stationary distribution. The reversibility implies that $A(x,y) \equiv \sqrt{\frac{\pi(x)}{\pi(y)}} P(x,y)$ is a symmetric matrix, and therefore has orthonormal eigenvectors $f_k(x)$ with corresponding eigenvalues $\lambda_k$. The eigenvectors for $P(x,y)$ will then be $g_k(x) = f_k(x)/\sqrt{\pi(x)}$. Now we want to expand the initial distribution $p_0(x) = \delta_{x,x_0}$ in terms of $g_k(x)$,

$$p_0(x) = \sum_k \alpha_k g_k(x), \alpha_k = \pi(x_0)g_k(x_0),$$

which can be verified using the orthogonality of $f_k$. After $t$ steps, the probability distribution would be

$$p_t(x) = \sum_k \alpha_k \lambda_k^t g_k(x),$$

and the probability for state $x_0$ would be

$$p_t(x_0) = \sum_k \alpha_k \lambda_k^t g_k(x_0) = \sum_k \pi(x_0) g_k(x_0)^2 \lambda_k^t,$$

which is a monotonically non-increasing function of $t$ given $\lambda_k$ are all nonnegative. □

**Proposition 3** (Restatement of Lemma 2). *After a label "N" reaches the target vertex $y$, the probability for having an "N" on $y$ will remain $\Omega(1)$.*

*Proof.* We again only keep track of the label "N" closest to $y$. When it is at $y$, there is a probability of $\frac{d_y}{E}\frac{1}{d^2+1}$ that an "I" is left on $y$ and the closest "N" becomes one of the neighbors of $y$. Here, $d_y$ is the degree of $y$. Otherwise, suppose it is at a vertex $u$ with degree $d_u$. There is at least one neighbor of $u$ that is one step closer to $y$, and other

neighbors are at most one step farther. This corresponds to a probability of $\frac{1}{E}\frac{d^2}{d^2+1}$ to move one step closer to $y$, and $\frac{d_u-1}{E}\frac{1}{d^2+1}$ to move one step farther.

Then we only need to consider a random walk on the states $\{0, 1, \ldots, d_{\max}\}$, where $d_{\max}$ is the largest possible distance from $y$. At state 0 there is a chance of $\frac{d^2}{E(d^2+1)}$ moving to 1, and at other state $k$ the probabilities for moving to $k+1$ and $k-1$ are $\frac{d^2-1}{E(d^2+1)}$ and $\frac{d^2}{E(d^2+1)}$ respectively. From [3] we can get the probability for state 0 in the stationary distribution, which is

$$\frac{1}{1+\sum_{t=0}^{d_{\max}-1}\left(\frac{d^2-1}{d^2}\right)^t} > \frac{1}{d^2+1},$$

which is $\Omega(1)$. Note that at large $E$, the probability for staying at the same state is larger that $\frac{1}{2}$, which means that all eigenvalues are positive. Also the random walk on a finite chain is reversible, so by Lemma 1, the probability for having "N" on $t$ remain $\Omega(1)$. □

## Appendix E: OTOC for the high degree case

### 1. Examples with long OTOC time

In this section we describe two examples of a high-degree graph where the OTOC time is asymptotically larger than the graph diameter. One example is the star graph where this phenomenon was previously observed by Lucas [4]. A second example is a binary tree with high but constant degree. We summarize the parameters of these two examples in the following table.

| Graph | Diameter | OTOC time | Source |
|---|---|---|---|
| $n$ qubits in star graph | 2 | $O(\log n)$ | [4] |
| tree with degree $z \gg d^2$ | $O\left(\frac{\log n}{\log z}\right)$ | $\tilde{O}\left(n^{1-\frac{\ln(d^2)}{\ln z}}\right)$ | this section |

In this section we will sketch proofs of both of the claimed OTOC times. First we give a simplified explanation of the $O(\log n)$ OTOC time for the star graph. (A more precise but also more complicated argument was given in [4].) In the star graph there are $n-1$ vertices each connected to a central vertex.

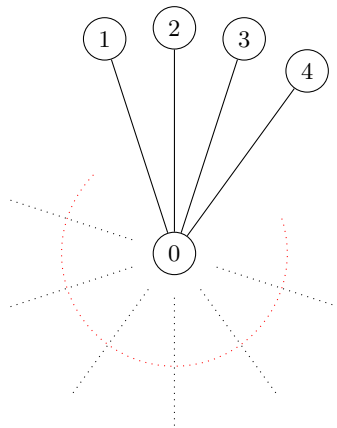

FIG. 2. Star Graph

Consider the OTOC between Paulis on two of the non-central vertices, say 1 and 2. We begin with a single N on vertex 1 and I's everywhere else. Vertices 0 and 1 interact after expected time 1, and when they interact there is a 4/5 chance that vertex 0 ends up with an N. Thus the expected time for vertex 0 to turn to N is 5/4 and after this happens the expected number of non-central N's is 3/4.

After this, vertex 0 rapidly interacts with the other $n-1$ vertices. Each time there is a 3/5 chance that the other vertex will turn from I to N, a 1/5 chance that the other vertex will remain I (or change from N to I, but consider for simplicity the early stages when most vertices are I), and a 1/5 chance that the other vertex will become N and vertex 0 will revert to I. The number of N's created here before vertex 0 turns back to I is given by a geometric random variable with expectation 3, and the time this takes has expectation $\approx \frac{4}{n}$. For simplicity, we will ignore the fluctuations and assume that the process creates exactly three N's and takes time exactly $\frac{4}{n}$.

Now there (on expectation) 3.75 N's in non-central vertices and I's elsewhere. Again there is a long wait until the central vertex turns back to N, but now the expected wait is only $\frac{5}{4} \cdot \frac{1}{3+\frac{3}{4}}$. (We pretend, crudely, that there are exactly 3.75 N's in non-central vertices.) Again (on average) 3 more non-central vertices turn to N and this again takes expected time $\frac{4}{n}$.

The time for $O(n)$ of the non-central vertices to turn to N is thus dominated by the waits for the central vertex to turn back to N which add up to

$$\frac{5}{4}\left(1 + \frac{1}{3+\frac{3}{4}} + \frac{1}{2 \cdot 3 + \frac{3}{4}} + \frac{1}{3 \cdot 3 + \frac{3}{4}} + \cdots + \frac{1}{O(n) \cdot 3 + \frac{3}{4}}\right) = O(\log n). \tag{E1}$$

This is not quite rigorous as we have replaced random variables with their expectations in several places, but otherwise is strong evidence for the OTOC time of $O(\log n)$ for the star graph. A more rigorous argument is given in Section 5 of [4]. There it is also observed that dynamics based on random Hamiltonians give an asymptotically different OTOC time.

We can obtain a stronger separation between diameter and OTOC time by considering a tree. Here we need the degree $z$ only to be a sufficiently large constant and can obtain an OTOC time of $n^{1-\delta}$ with $\delta \to 0$ as $z \to \infty$. However the diameter is also asymptotically growing as $O(\log n)$, unlike the star graph where the diameter is only 2.

Suppose that degree $z \gg d^2 \gg 1$. Consider a tree of depth $H$, so that there are $(z-1)^h$ vertices at level $h$ for $h = 0, \ldots, H$. The total number of vertices is thus $n = 1 + (z-1) + (z-1)^2 + \cdots + (z-1)^H = z^{H+1} - 1$. Each vertex with $h > 0$ has one parent and each vertex with $h < H$ has $z - 1$ children.

Consider an N at vertex $v$, with height $h$. Suppose that $x(z-1)$ of its children are in state $N$ and $(1-x)(z-1)$ are in state $I$ for some $x \in [0, 1]$. We approximate the dynamics as follows.

- The parent of $v$ becomes $N$ at rate $1 - 1/(d^2+1) \approx 1$, regardless of its current state.

- A child of $v$ turns into $N$ at a rate $(z-1)(1-1/(d^2+1)) \approx z$, again regardless of its current state. This means that $\frac{dx}{dt}$ has an expected contribution of $(1-x)$.

- A child of $v$ currently in state $N$ turns to $I$, at rate $x(z-1)/(d^2+1) \approx xz/d^2$, yielding a contribution to $\frac{dx}{dt}$ of $-x/d^2$.

- $v$ itself turns to $I$ at a rate of $(z-1)/d^2 \approx z/d^2$. (At the same time, one of $v$'s children turns to $N$, but we can neglect this.)

Overall we find that $x$ evolves according to

$$\frac{dx}{dt} \approx 1 - (1 + 1/d^2)x, \tag{E2}$$

which asymptotically approaches $(1+1/d^2)^{-1} \approx 1 - 1/d^2$. However, $v$ will turn to $I$ before this happens, typically after time $d^2/z$. During this time, the probability of $v$'s parent turning to $N$ is $\approx z/d^2$.

Once $v$ turns to $I$, its children can turn it back to $N$. This happens at rate $\approx xd$.

First, let us analyze more carefully the dynamics of $x$. If $h \leq H - 2$ then the children of $v$ will themselves undergo the same process. In particular, they will turn to $I$ at rate $z/d^2$. This means Eq. (E2) is modified to become

$$\frac{dx}{dt} \approx 1 - (1 + 1/d^2 + z/d^2)x \approx 1 - \frac{z}{d^2}x, \tag{E3}$$

Now $x$ will saturate at $x = d^2/z$.

Putting this together, the N's will "fall" to lower levels more quickly than they can climb or be replaced.

By the time they've fallen to the base, the N's correspond to the leaves of a subtree $S$ with branching factor $d^2$. Then they begin slowly climbing back up.

Let us start at level $H$. Fix a vertex $v \in S$ at height $H - 1$. After time $1/d^2$, $v$ will turn to N. Then it will fall back down again after time $d^2/z$, during which it will have created another $d^2$ N's at level $H$. Now there are $2d^2$ N's below

$v$. This means it will take time $1/2d^2$ to turn $v$ back to N. Again it falls down after time $d^2/z$, creating another $d^2$ N's. This process continues $z/d^2$ times, for a total time of

$$\frac{1}{d^2}\left(1 + \frac{1}{2} + \ldots + \frac{1}{z/d^2}\right) + \frac{z}{d^2} \cdot \frac{d^2}{z} \approx 1 + \frac{\ln(z)}{d^2}.$$

This happens in parallel for each height-$(H-1)$ vertex in $S$. These vertices no longer fall because most of their children are (and remain) equal to N.

Now consider a height-$(H-2)$ vertex in $S$, again called $v$. By the previous arguments, it has $d^2$ children that remain mostly permanently set to N. So $v$ will become N after time $1/d^2$. However, after time $d^2/z$ it will fall down to level $H-1$, creating $d^2$ N's at level $H-1$ in the process; call this set $A$. These in turn will fall to level $H$ where they create $d^4$ N's. Again it takes time $1 + \ln(z)/d^2$ to cause the vertices in $A$ to turn permanently into N's. After this it takes time $1/2d^2$ to again turn $v$ to N, etc. There are again $z/d^2$ rounds, and focusing only on the $1 + \ln(z)/d^2$ contribution to the time for one round, we obtain a total time of

$$\frac{z}{d^2}(1 + \ln(z)/d^2).$$

Now we move on to the height-$(H-3)$ part of $S$. Each vertex there again turns to N and then falls to the bottom of the tree $z/d^2$ times. But now climbing back up takes time

$$\approx \left(\frac{z}{d^2}\right)^2 (1 + \ln(z)/d^2).$$

We conclude that the entire OTOC time is

$$\approx \left(\frac{z}{d^2}\right)^{H-1}\left(1 + \frac{\ln(z)}{d^2}\right). \tag{E4}$$

Using that $H + 1 = \ln(n+1)/ln(z)$, the OTOC time is

$$\approx \left(\frac{d^2}{z}\right)^2 \left(\frac{z}{d^2}\right)^{H+1} = \left(\frac{d^2}{z}\right)^2 (n+1)^{\frac{\ln(z/d^2)}{\ln(z)}} = \left(\frac{d^2}{z}\right)^2 (n+1)^{1-\frac{\ln(d^2)}{\ln(z)}}. \tag{E5}$$

Thus by taking $z = d^{1/\delta}$ we obtain OTOC scaling as $n^{1-\delta}$ which is much larger than the diameter of $O(H) = O(\ln(n)/\ln(z)) = O(\delta \ln(n)/\ln(d))$.

## 2. Examples with short OTOC time

For a complete graph with $V$ vertices, $\Theta(V \log V)$ edges should be picked in order for the OTOC to be $\Theta(1)$, according to [5]. By Proposition 1, the OTOC saturation time is $\Theta(\log(V)/V)$. In contrast, the distance between any pair of vertices is 1.

Another example is that the OTOC saturation time on any graph can be scaled down by a factor of $\Theta(1/k)$ if we replace each edge by $k$-fold parallel edges. Alternatively we can replace each edge $(u,v)$ by $2k$ edge $(u,z_i)$ and $(z_i,v)$ for $1 \leq i \leq k$, where $\{z_i\}_{i=1}^k$ are newly added vertices (there are $2kE$ new vertices in total, $k$ for each edge). Since $k$ could be arbitrary, this violates the lower bound we obtained at finite degree.

## Appendix F: Proof for Theorem 2

**Proposition 4** (Restatement of Theorem 2). *Let $G$ be a graph with $V$ vertices and $E$ edges, and suppose the degree for each vertex is $O(1)$. Then for any pair of vertices $x$ and $y$, $\tau_{\text{OTOC}}^{(x,y)} = \Omega(D(x,y))$ with probability $1 - e^{-\Theta(D(x,y))}$, where $D(x,y)$ is the distance between $x$ and $y$.*

*Proof.* Consider a Markov chain $M'$ on all configurations in which every vertex of $G$ is assigned labels "N" or "I". Initially $x$ is assigned "N" and other vertices are assigned "I". In each step a random edge is picked and the vertices it connects are changed to "N" if at least one of the vertices was previously assigned "N". Using a simple coupling argument the number of steps needed for $y$ to be assigned "N" gives a lower bound for $\tau_{\text{OTOC}}^{(x,y)}$. To show a lower bound

for the number of steps needed in $M'$, we prove an even stronger statement that with probability at least $1 - e^{-\Theta(d)}$, all of the vertices with distance $d$ to vertex $x$ (denoted by $S_d$) have label "I" after $\Theta(E \cdot d)$ steps.

We start with a Chernoff-type bound for a sum of $k$ i.i.d. geometric variables. Let $X = X_1 + \ldots + X_k$ where $X_i$ has probability of $p(1-p)^{m-1}$ for value $m$, $m \geq 1$. For any $t > 0$, we have $\mathbb{E}e^{-tX_i} = \frac{p}{e^t - (1-p)}$, and therefore

$$\mathbb{E}e^{-tX} = \left(\frac{p}{e^t - (1-p)}\right)^k.$$

The mean for $X$ is $\mu \equiv \mathbb{E}X = \frac{k}{p}$. By Markov inequality,

$$\begin{aligned}
\Pr[X \leq \lambda\mu] &= \Pr[e^{-tX} \geq e^{-t\lambda\mu}] \\
&\leq e^{t\lambda k/p} \mathbb{E}e^{-tX} \\
&= \exp\left[k\left(\ln\frac{p}{e^t - (1-p)} + \frac{t\lambda}{p}\right)\right]
\end{aligned}$$

for any $t > 0$. When $\lambda > p$, we take $t = \ln\frac{(1-p)\lambda}{\lambda - p}$ and get $\Pr[X < \lambda\mu] \leq \exp(-kf(p, \lambda))$ where

$$f(p, \lambda) = -\ln\frac{\lambda - p}{1 - p} - \frac{\lambda}{p}\ln\frac{\lambda(1-p)}{\lambda - p}.$$

In Markov chain $M'$, the number of steps needed for spreading labels "N" along a path with length $d$ is a sum of $d$ geometric variables with $p = 1/E$. Let $C$ be the max degree of the vertices. From vertex $x$ to vertices in $S_d$, there could be at most $(C-1)^d$ paths with length $d$, so by union bound the probability that any vertex in $S_d$ gets label "N" after $\lambda dE$ steps is at most $(C-1)^d \exp(-df(p, \lambda))$. Knowing that $f(p, \lambda) = -\log\lambda + \lambda - 1 + O(p)$ can be arbitrarily large for small constant $\lambda$, we can pick $\lambda$ such that the probability becomes $e^{-\Theta(d)}$. By Proposition 1, the $\lambda dE$ steps here corresponds to $\Theta(\lambda d) = \Theta(d)$ time units. $\square$

## Appendix G: Variants of our model

As mentioned in the main text, one may be concerned that the slow entanglement saturation is due to the low chance of picking the cut edge in each step. Here we first consider a naturally modified model where the probability of picking the cut edge is much higher, albeit the OTOC and entanglement times remain separated. In this model we still apply Haar-random unitaries to vertices connected by randomly chosen edges. Instead of picking a random edge from all edges, in each step we pick a random edge from the left subtree (L), and then the cut edge (C), and finally a random edge from the rest of the edges (R). In this LCR model, the analysis for OTOC does not change by much and still has the $O(V \log V)$ upper bound. We would argue that the entanglement saturation time is still $\Omega(V^2)$. Consider the sequence of picking edges $L_{i_1} C R_{j_1} L_{i_2} C R_{j_2} \ldots$. The probability that the random edge on $L$ or $R$ is connected to the cut edge is $\Theta(1/V)$. So the probability that neither $R_{j_k}$ nor $L_{i_{k+1}}$ is connected to $C$ is $1 - \Theta(1/V)$, in which case $R_{j_k} L_{i_{k+1}}$ commute with the neighboring $C$, and the chain reads $\ldots L_{i_k} CCR_{j_k} L_{i_{k+1}} R_{j_{k+1}} \ldots$, where the two $C$ gates reduce to one effective $C$ gate. The probability that such commutation does not happen is $\Theta(1/V)$, which means that the expected number of effective $C$ gates in constant time steps is $\Theta(1/V)$. So the expected saturation time is $\Omega(V^2)$.

We can also consider two complete graphs connected by a cut edge. Let $V$ be the number of vertices, and the number of edges would be $E = \Theta(V^2)$. We still have a similar Markov chain describing the evolution of OTOC. In the LCR model for this graph, by [5] the complete graph on the left reaches equilibrium in $O(\log V)$ time. Then after a constant number of steps there will be a label "N" propagated to the right, and another $O(\log V)$ steps will be needed for OTOC to equilibrate. On the other hand, $\Omega(V)$ steps will be needed for entanglement. If a random edge is picked in each step, we expect the OTOC time scales to stay the same. In this case, though, the entanglement would take time $\Omega(E) = \Omega(V^2)$ to saturate.

Besides the two ways of picking random edges, we can also pick a random matching in each step (assuming the number of vertices on each side is even), followed by picking the cut edge. The entanglement entropy takes $\Omega(V)$ steps to saturate. For OTOC, the number of edges connecting a vertex with label "N" and a vertex with label "I" in a random matching should be proportional to the total number of "N", so we expect the number of vertices with "N" to grow exponentially. Therefore it would take $O(\log V)$ steps for an "N" to go across the middle edge, and a total

number of $O(\log V)$ steps for OTOC to saturate.



\end{document}